%% file: main.tex
\documentclass[11pt,letterpaper,final]{article}
\usepackage[papersize={8.5in,11in},margin=1in]{geometry}
\usepackage{microtype}
\usepackage{comment}
\usepackage{amsmath,amssymb,amsthm}
\usepackage{xcolor}
\usepackage{booktabs,tabularx}
\usepackage{caption}
\usepackage{algorithm}
\usepackage{placeins}
\usepackage[hidelinks]{hyperref}
\usepackage[indLines=true]{algpseudocodex}
\usepackage{tikz}
\usepackage[disable]{todonotes}

\usetikzlibrary{
  arrows.meta,
  decorations.markings,
  decorations.pathreplacing,
  positioning
}
\definecolor{classicblue}{RGB}{0,114,178}
\definecolor{classicgreen}{RGB}{0,158,115}
\definecolor{neutralgray}{RGB}{127,127,127}

\newtheorem*{theorem*}{Theorem}

\newtheorem*{lemma*}{Lemma}

\theoremstyle{definition}

\theoremstyle{remark}

\title{\textbf{Space-Efficient Hierholzer for Undirected Graphs}}

\author{
    Elena Grigorescu\thanks{University of Waterloo, Canada.
    \texttt{\{elena.grigorescu,\,s3shirazimofrad\}\,@\,uwaterloo.ca}}
    \qquad
    \stepcounter{footnote}
    Ziad Ismaili Alaoui\thanks{University of Liverpool, United Kingdom.
    \texttt{\{ziad.ismaili-alaoui,\,sebastian.wild\}\,@\,liverpool.ac.uk}}
    \qquad
    Tamio-Vesa Nakajima\thanks{Philipps-Universit\"at Marburg, Germany.
    \texttt{\{nakajima,\,wild\}\,@\,informatik.uni-marburg.de}}
    \\[0.5em]
    Shayan Shirazi Mofrad\footnotemark[1]
    \qquad
    Sebastian Wild\footnotemark[3]\,\,\footnotemark[4]
}

\date{}

\newif\ifshowcomments
\showcommentstrue
\ifshowcomments
\newcommand{\enote}[1]{{\color{green}[{\small Elena: \bf #1}]\marginpar{\color{red}*}}}
\newcommand{\snote}[1]{{\color{purple}[{\small Shayan: \bf #1}]\marginpar{\color{red}*}}}
\newcommand{\tnote}[1]{{\color{orange}[{\small Tamio: \bf #1}]\marginpar{\color{purple}*}}}
\newcommand{\znote}[1]{{\color{teal}[{\small Ziad: \bf #1}]\marginpar{\color{teal}*}}}

\else
\newcommand{\enote}[1]{}
\newcommand{\snote}[1]{}
\newcommand{\tnote}[1]{}
\newcommand{\znote}[1]{}
\newcommand{\todo}[1]{}

\fi

\newcommand{\bigO}{\ensuremath{\mathrm{O}}}

\begin{document}
\maketitle

\todo[inline,caption=global]{\begin{minipage}{.9\linewidth}
    GLOBAL TODO (by priority)
    \begin{itemize}
        \item end-to-end proofread ;)
              Minor grammar improvements
        \item harmonize notation
        \begin{itemize}
            \item Eulerian tour (not cycle/circuit) for global problem
        \end{itemize}
        
    \end{itemize}
\end{minipage}}

\begin{abstract}
We present a simple linear-time algorithm that outputs an Eulerian tour of an
undirected multigraph with $n$ vertices and $m$ edges, if one exists, in $\bigO(m)$ time and using
$\bigO(n)$ words of working memory.
The input is given as read-only adjacency lists, and the output is written to an append-only stream in traversal order. Our algorithm first finds a sparse spanning circuit (a \emph{skeleton}),
then traverses the circuit step-by-step, repeatedly outputting further circuits rooted at the current vertex. This solves a problem left open by Ismaili Alaoui, Plump,
and Wild (SOSA~2026): their space-efficient variant of Hierholzer's algorithm
handles general \emph{directed} multigraphs, but it is unclear how to generalize it to general \emph{undirected} multigraphs. Our result completes the picture in the read-only model for space-efficient output of Eulerian tours.
\end{abstract}

\section{Introduction}
\label{s:introduction}

We study the well-known problem of computing an \emph{Eulerian tour} of an undirected (multi)graph, i.e.~a tour containing all the edges of a given multigraph. Our work is motivated by the recent results of Ismaili Alaoui, Plump, and Wild~\cite{IsmailiPW26}, who obtain an $\bigO(m)$-time algorithm using $\bigO(n)$ words of working memory for \emph{directed} multigraphs; however, it is unclear how to generalize this algorithm to \emph{undirected} multigraphs.
Our main result is the following (illustrated in Figure~\ref{fig:skeleton-overview}).

\begin{theorem*}[Main]
Let $G$ be an undirected Eulerian multigraph with $n$ vertices and $m$ edges, represented by read-only, unsorted adjacency arrays without edge identifiers. There exists an algorithm that outputs an Eulerian tour of $G$ sequentially on an append-only output stream in $\bigO(m)$ time using $\bigO(n)$ words of working memory.
\end{theorem*}

Our algorithm can also easily be extended to Eulerian walks, as described in  Section~\ref{s:conclusion}. Finding an Eulerian tour is one of the foundational problems of graph theory with applications in genome assembly~\cite{PevznerTW01}, route inspection~\cite{Orloff74}, and continuous tool-path planning~\cite{Yamamoto22}. Hierholzer is often credited with first proposing a set of instructions to systematically construct Eulerian tours in undirected multigraphs, and his results were posthumously communicated by Wiener \cite{Hierholzer1873}. Needless to say, his work did not specify how to implement this algorithm in a time- and space-efficient way for modern computers.

\input{ear-decomp-fig}

\paragraph{Previous Work.}
Direct implementations of Hierholzer's algorithm run in linear time, but either maintain the partially constructed tour as a dynamic list or store the stack of an edge-centric DFS, together with information recording which edges have already been used, so their working memory grows with the number of edges. 
Implementations in widely used software frameworks also follow this paradigm (see~\cite{IsmailiPW26} for more discussion).

Several algorithms reduce this space requirement. For example, Hagerup, Kammer, and Laudahn~\cite{HagerupKL19} compute an Euler partition of an undirected graph in $\bigO(m)$ time using $\bigO(m)$
bits of space. Their bound is incomparable to ours:
$\bigO(m)$ bits is smaller for sparse graphs, while
$\bigO(n)$ words (which amounts to $\bigO(n \lg m)$ bits) is smaller whenever $m=\omega(n\lg n)$.
Glazik, Schiemann, and Srivastav~\cite{GlazikSS23} give a one-pass
streaming algorithm that uses $\bigO(n\lg n)$ bits; however, it outputs
an implicit successor representation of the tour rather than the tour
itself written in order, and no linear runtime bound is established.

For \emph{directed} (multi)graphs, the algorithm of Ismaili Alaoui, Plump, and
Wild~\cite{IsmailiPW26} uses a reverse traversal, reserving one distinguished
incoming edge per vertex to be the last edge on which the traversal
backtracks through that vertex. This requires only a constant amount of
information per vertex. 
It is unclear how to generalize their algorithm to undirected graphs.

\paragraph{Our Results.}
We here give a rather different approach to implementing Hierholzer's algorithm, 
with the same space and time complexity, and in the same model as~\cite{IsmailiPW26}.
The results are orthogonal;
our new method cannot work with directed graphs,%
\footnote{%
     At least as it is, our algorithm has no hope to be adapted to the directed case due to a result of Colòn and Urschel~\cite{ColonUrschel2024}: there exists an infinite family of directed Eulerian graphs on $n$ vertices and $\Theta(n^{3/2})$ edges that do not contain a proper subgraph that is both Eulerian and a spanning subgraph.
} but solves the open problem for undirected (multi)graphs.
Our algorithm can easily be extended to handle Eulerian paths.
Table~\ref{table-refs} provides a quick comparison of the results discussed
above.

Together with the directed algorithm of Ismaili Alaoui, Plump, and
Wild~\cite{IsmailiPW26}, our result completes the picture in the read-only
adjacency-array model considered here: Eulerian tours of both directed and
undirected multigraphs can be output in order in $\bigO(m)$ time using
$\bigO(n)$ words of working memory.

\begin{table}
\centering
\small
\setlength{\tabcolsep}{4pt}
\renewcommand{\arraystretch}{1.15}
\begin{tabularx}{\textwidth}{
    @{}
    >{\raggedright\arraybackslash}p{0.3\textwidth}
    >{\raggedright\arraybackslash}X
    >{\raggedright\arraybackslash}p{0.15\textwidth}
    >{\raggedright\arraybackslash}p{0.20\textwidth}
    @{}
}
\toprule
\textbf{Algorithm} &
\textbf{Setting} &
\textbf{Time} &
\textbf{Space} \\
\midrule

Standard Hierholzer &
Directed or undirected multigraphs &
$\bigO(m)$ &
$\bigO(m\lg n)$ bits \\
\midrule

Hagerup et al.~\cite{HagerupKL19} &
Undirected graphs; general representation &
$\bigO(m)$ &
$\bigO(m)$ bits \\
\midrule

Glazik et al.~\cite{GlazikSS23} &
One-pass streaming model &
not established &
$\bigO(n\lg n)$ bits \\
\midrule

Ismaili Alaoui et al.~\cite{IsmailiPW26} &
Directed multigraphs &
$\bigO(m)$ &
$\bigO(n\lg m)$ bits \\
\midrule

\textbf{This paper} &
Undirected multigraphs&
$\bigO(m)$ &
$\bigO(n\lg m)$ bits \\
\bottomrule
\end{tabularx}
\caption{Comparison of algorithms for computing Eulerian tours.}
\label{table-refs}
\end{table}

\section{Preliminaries}
\label{s:model}

\paragraph{Definitions.}
We will assume our graphs have at least two vertices. We disallow isolated vertices. We allow both parallel edges and loops in our graphs; a loop counts twice towards the degree of a vertex.
We define $+$ and $-$ on multigraphs by addition and subtraction of edge multiplicities.
(All our graphs will share the same vertex set.)

A \emph{walk} in a (multi)graph $G$ is a sequence of edges $(u_1, v_1), \ldots, (u_k, v_k)$ where $v_i = u_{i + 1}$ for $i = 1, \ldots, k - 1$. Our walks are \emph{edge simple}: they cannot reuse edges. A \emph{circuit} is a walk where $u_1 = v_k$. We say that a circuit is \emph{spanning} if it contains all the vertices of~$G$. An \emph{Eulerian tour} is a circuit that contains every edge of $G$ exactly once. We say that a multigraph is \emph{even} if all its degrees are even, and \emph{Eulerian} if it admits an Eulerian tour.
A \emph{skeleton} of $G$ is a spanning circuit $S$ containing at most $2n$ edges.

Recall Euler's theorem: $G$ is Eulerian if and only if it is even and connected. Furthermore, recall that Hierholzer's algorithm can find an Eulerian tour for any Eulerian graph in $O(n + m)$ time and space. We will use the following simple corollary freely: Any even graph has Eulerian connected components, and we can find a partition of the edges of an even graph into circuits in $O(n + m)$ time and space, by applying Hierholzer's algorithm to each connected component. Recall also the Handshake lemma: in any graph $G$, the sum of the degrees of all vertices is even. Note that Euler's theorem and the Handshake lemma hold for graphs with parallel edges and loops, with our convention for counting degrees with loops.

\paragraph{Computational Model.}
The vertices are labeled $1,\ldots,n$, and for each vertex $v$, the input contains a read-only adjacency array $A(v)$. 
Every non-loop edge $uv$ appears in the input as a copy of $u$ in $A(v)$ and a copy of $v$ in $A(u)$ --- loops appear once. 
Parallel edges and self-loops are allowed, with the different instances of a parallel edge being indistinguishable in the adjacency lists of their endpoints.
We work in the standard word RAM model: working memory is measured in words, where a vertex, multiplicity, or adjacency list position occupies one word, and reading or updating one word takes constant time. 
The output is an Eulerian tour in which each edge is traversed exactly once. 
The read-only input arrays and the append-only output stream are not counted as working memory.

\paragraph{Multigraph Subtraction.} Our algorithm repeatedly manipulates multigraphs of form $G - H$, where $G$ is the input graph (and hence unmodifiable), and $H$ is some \emph{sparse} subgraph of $G$. 
To simulate access to $G-H$ using $O(n)$ working space, we proceed as follows. 
First, explicitly store the adjacency lists of $H$. We next reorder these so that the adjacency list of each vertex is a (not necessarily contiguous) subsequence of the adjacency list in $G$. (Do this adjacency list by adjacency list. For a particular vertex $i$, count the frequency of each edge $ij$ in $H$ in a frequency table. Then iterate through the neighbours of $i$ in $G$, and write down the new adjacency list in $H$ by outputting edges that still have nonzero frequency, decrementing them in the frequency table after outputting them. Note that this decrementing resets the frequency table to contain only zeroes, avoiding a full reset which could take $O(n)$ per vertex.)

With this in mind, we can simulate iterating through the adjacency lists of $G - H$ by iterating in parallel through those of $G$ and $H$, skipping one copy of each edge from $H$ in $G$. Compared to traversing $G$, this incurs an overhead depending on the size of $H$ and the number of times we traverse $G - H$, but $H$ will be small and all of our algorithms will traverse each adjacency list of $G - H$ at most $O(1)$ times, and thus this is negligible.  

\paragraph{Forest Lemma.}
We use a simple folklore lemma, which we prove for completeness; graph theorists will recognize this as computing the unique \emph{$T$-join} within a spanning forest of $G$ for $T$ the set of vertices of $G$ with odd degree.
\begin{lemma*}
    Given a multigraph $G$ with $n$ vertices and $m$ edges, we can find a forest $F \subseteq G$ in $\bigO(m)$ time and $\bigO(n)$ space, such that $G-F$ is even.
\end{lemma*}
\begin{proof} 
    Solve each component separately, so assume $G$ is connected. Run a DFS on $G$, computing a spanning tree $T$. For each vertex in a postorder traversal of $T$, compute the parity of the total $G$-degree of its subtree in $T$; if odd, add its parent edge to $F$. This is well-defined since the parity at the root is even by the Handshake lemma. A short parity check shows $G-F$ is even, $F$ is a subgraph of $T$ and thus a forest, and DFS gives the claimed
    complexity.
\end{proof}

\section{Algorithm}

We now present and prove the main result of the paper. Our algorithm has two phases. The setup phase computes a skeleton $S$. The main phase uses $S$ as a \emph{roadmap}
to compute an Eulerian tour of $G$: it traverses $S$ step by step, outputting circuits rooted at the current vertex, and using $S$ to ensure that we can eventually get back to the start vertex (cf.\ Figure~\ref{fig:skeleton-overview}).

\subsection{Overview}

\paragraph{Setup phase.}
We first construct the skeleton $S$. Run a DFS on $G$ to obtain a spanning tree~$T$, and apply the Forest Lemma to
$G-T$.
This produces a forest $F\subseteq G-T$ such that $G-(T+F)$ is even. Since $G$ is even, $T+F$ is also even; moreover, it is connected
because it contains the spanning tree $T$. Therefore, $T+F$ is spanning and Eulerian. Since both $T$ and $F$ are forests, $T + F$ has at most $2n - 2$ edges,
so we can find a skeleton of $G$ by computing an Eulerian tour of $F + T$ in $O(n)$ time and space, using standard Hierholzer's algorithm. The entire setup takes $\bigO(m)$ time and uses $\bigO(n)$ words of working memory.

Let $v_1,\ldots,v_n$ be the vertices in the order
they first appear in $S$
starting from an arbitrary start vertex $v_1$.
By explicitly storing both the sequence $v_1, \ldots, v_n$ and the inverse map $v_i \mapsto i$ (which takes $\bigO(n)$ words), we may assume without loss of generality that $v_1 = 1, \ldots, v_n = n$.
Split $S$ into walks $W_1,\ldots, W_n$ where $W_i$ starts at $i$ and ends at $i+1$; $W_n$ starts at $n$ and ends at 1.

\paragraph{Main Phase.} We will sequentially construct a sequence of circuits $C_1, \ldots, C_n$. Circuit
$C_i$ will start and end at vertex $i$.
Our Eulerian tour will then be the
interleaving 
\[
	C_1, W_1,\, C_2, W_2,\,\ldots,\,C_{n-1},W_{n-1},\,C_n,W_n\,.
\] 
The circuits will have the following property. Define $G_0, G_1, \ldots, G_n$ as
subgraphs of $G - S$ with the same vertex set,%
\footnote{%
    Formally speaking, we should write $G - H$ where $H$ is the sub(multi)graph of $G$ formed from the edges occurring in $S$. For brevity, we identify subgraphs and edge sets, and use tours as edge sets.
} 
and where $G_i$ contains all the
edges $uv$ where $\min(u, v) \leq i$. 
(Of course, $G_0$ has no edges.)
The key property (which follows from a simple induction, see below) will be the
following: there exist forests $F_0 \subseteq G_0, F_1 \subseteq G_1, \ldots, F_n \subseteq G_n$
such that for every $i = 0, \ldots, n$,
\begin{equation}\label{eq:invariant}
    C_1 + \cdots + C_i + F_i \;=\; G_i.
\end{equation}
In other words, $G_i$ is partitioned into $C_1, \ldots, C_i$ and $F_i$. Note that this partition is always possible since, by the Forest Lemma, there always exists an $F_i$ such that $G_i-F_i$ is even, and any such graph can be decomposed into circuits. At every step $i$, we will not explicitly remember $C_1, \ldots, C_i$; rather,
only $F_i$, which is a forest, i.e.~sparse, and thus fits within our memory
constraints. Note furthermore that $F_n$ \emph{must} be the empty graph: since $C_1 + \cdots + C_n$ and $G_n$ have even degrees, $F_n$ does too, and $F_n$ is a forest; the only such graph
is the empty graph. Hence~\eqref{eq:invariant} implies we output a full Eulerian tour. Note also that~\eqref{eq:invariant} obviously holds for $i = 0$.

\subsection{Naive Implementation}

We first describe a naive way to implement this, which degrades to quadratic time. Suppose we have just
finished step $i - 1$, having computed and stored $F_{i - 1}$, and output $W_{i - 1}$ if it exists. We now want to output
$C_{i}$. To do this, let $N_i$ be the edge neighbourhood of $i$ in $G_i$, or equivalently $N_i = G_i - G_{i - 1}$. We
apply the Forest lemma to $F_{i - 1} + N_i$, to create the forest $F_i$.
Note
that by construction $F_i$ is a forest, and $(F_{i - 1} + N_i) - F_i$ 
has even degrees; thus we can partition the edges of $(F_{i - 1} + N_i) - F_i$
into one circuit per connected component. But note that the edges therein are taken from the forest
$F_{i - 1}$ together with some arbitrary edges $N_i$, all incident at $i$. Thus
any circuit must contain $i$, and in fact our partition consists of a single circuit containing $i$, which we take to be $C_i$ (unless $(F_{i - 1} + N_i) - F_i$ has no edges, in which case $C_i$ is also empty). To see why~\eqref{eq:invariant} is maintained, note that
\[
C_i + F_i - F_{i - 1}
=
(F_{i - 1} + N_i - F_i)+ F_i - F_{i - 1} = N_i = G_i - G_{i-1}.
\]
Adding this equation to~\eqref{eq:invariant} for $i - 1$ yields it for $i$. 

\subsection{Linear Time Implementation}

We first describe a method to implement the algorithm in linear time for simple graphs, and then discuss multigraphs.

\paragraph{Simple Graphs.}The previous algorithm has quadratic runtime since the
computation of $C_i$ takes (up to) linear time, and we do this $n$ times. We first give a simplified linear-time algorithm, albeit one that only works for graphs \emph{without parallel edges}. The trick to reduce the runtime is to produce several of these circuits at once. 
Suppose we have just finished step $i - 1$. 
Now, compute $j$ such that the number of edges within $N_i, \ldots,
N_j$ is at most $2n$, and is either no less than $n$, or else $j = n$.
This is possible since each new $N_k$ adds at
most $n-1$ new edges.

Apply the Forest Lemma to $F_{i-1}+N_i+\cdots+N_j$ to obtain $F_j$. Now look
at $(F_{i-1} + N_i + \cdots + N_j) - F_j$. This graph has only even degrees, and
as before it contains only edges from the forest $F_{i - 1}$ and some arbitrary
edges $N_i, \ldots, N_j$, all incident at one of $i, \ldots, j$. Thus we can partition its
edges into a set of circuits, and each such circuit must touch at least one of $i,
\ldots, j$. Interleaving these circuits allows us to create $C_i, \ldots, C_j$, and we now output $C_i, W_i, \ldots, C_j, W_j$.

\paragraph{Parallel Edges and Loops.}
Finally, we explain how to treat parallel edges in the previous algorithm. Rather than selecting $N_i, \ldots, N_j$ to have between $n$ and $2n$ edges, we select them to contain between $n$ and $2n$ \emph{distinct} edges. If an edge has odd multiplicity, replace all copies by one virtual edge; if it has even multiplicity, keep one copy and replace the rest by one virtual edge. Expand virtual edges when writing to output. Each distinct edge produces at most two virtual edges, so every batch contains $O(n)$ virtual edges.

Iteratively count how many distinct edges are within some $N_k$ for $k = i, \ldots$, and furthermore output a list of virtual edges for $N_k$, together with how many original edges they represent. To compute this for $N_k$, we assume that we have an array $\operatorname{count}[1..n]$, which initially contains only zeros. We now iterate through the adjacency list of $k$ within $G-S$, incrementing $\operatorname{count}[v]$. After this scan, $\operatorname{count}[v]$ gives the multiplicity of $kv$.
Then, by traversing the adjacency list of $k$ yet again, we can count distinct edges, output the virtual edge lists, and reset $\operatorname{count}$ to zero; we ought only add virtual edges within $N_k$, i.e.~where $k \leq v$.%

\paragraph{Analysis.}
First, note that constructing batches of edges, and outputting lists of virtual edges, takes $O(m + n)$ time and $O(n)$ space. Once we have the batches, since they consist of only $\bigO(n)$ virtual edges at a time, processing batches takes $\bigO(n)$ memory;
furthermore, processing one batch takes $\bigO(n)$ time, but we
only process $\bigO(m / n)$ batches, since each batch (except the last) is constructed to have
$\Omega(n)$ edges. Hence the total time complexity is $\bigO(m)$.%
\footnote{%
	Note that strictly speaking, our improved algorithm only constructs some subset of the
    $F_1, \ldots, F_n$ forests, namely those $F_i$ at batch boundaries. This
    does not affect correctness, since our invariant holds for those $F_i$ that we do construct; we certainly reach $F_n$ in the end, and
	$F_n$ must still be the empty graph as before. 
}

\section{Conclusion}
\label{s:conclusion}
We have given a linear-time algorithm whose working memory depends on
the number of vertices rather than the number of edges. The Forest lemma
is used twice: first to turn a spanning tree into the skeleton,
and then to keep only a forest of unfinished edges between consecutive batches.
The skeleton tour $S$ can be seen as a central tour that is sequentially written during a traversal while, in between, incident edges on that traversal are grouped and any emerging circuit is stitched along the walk through the skeleton. We have provided multiple implementations for different cases (simple graphs, multigraphs, etc.) and have shown correctness for each of these.

We note in passing that our algorithm is easy to modify for Eulerian walks. In particular, simply replace the skeleton \emph{circuit} with a skeleton \emph{walk}, starting and ending at the odd parity vertices. This can be found the same way as the procedure we outlined earlier.

\paragraph{Acknowledgments.}
We would like to thank the organizers of the \href{https://sites.google.com/view/zar2026}{2026 Romanian Algorithm Days}, where two of the authors first had some preliminary discussion about this project.
We also thank Andrei Feodorov and Alireza Kaviani for helping us with proof-reading the draft.
The authors used generative AI assistance in preparing the visualizations and for literature search. The authors assume responsibility for all content.

\bibliographystyle{alpha}
\bibliography{bibliography.bib}

\end{document}

%% file: ear-decomp-fig.tex
\providecolor{classicorange}{RGB}{230,159,0}
\providecolor{classicpurple}{RGB}{204,121,167}
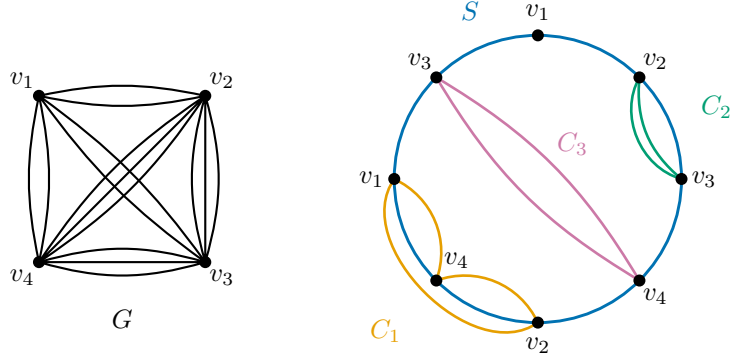
\begin{figure}
\centering
\begin{tikzpicture}[
    vtx/.style={circle, fill=black, inner sep=1.6pt},
    edge/.style={draw=black, line width=0.8pt},
    skel/.style={draw=classicblue, line width=1.1pt},
    circ/.style={draw=#1, line width=1.0pt},
    lbl/.style={font=\small}
]
    \coordinate (g1) at (-6.6, 1.1);
    \coordinate (g2) at (-4.4, 1.1);
    \coordinate (g3) at (-4.4,-1.1);
    \coordinate (g4) at (-6.6,-1.1);

    \draw[edge] (g1) to[bend left=12]  (g2)   
                (g1) to[bend right=12] (g2);
    \draw[edge] (g2) to[bend left=16]  (g3)   
                (g2) to[bend left=0]   (g3)
                (g2) to[bend right=16] (g3);
    \draw[edge] (g3) to[bend left=16]  (g4)   
                (g3) to[bend left=0]   (g4)
                (g3) to[bend right=16] (g4);
    \draw[edge] (g4) to[bend left=12]  (g1)   
                (g4) to[bend right=12] (g1);
    \draw[edge] (g1) to[bend left=9]   (g3)   
                (g1) to[bend right=9]  (g3);
    \draw[edge] (g2) to[bend left=11]  (g4)   
                (g2) to[bend left=0]   (g4)
                (g2) to[bend right=11] (g4);

    \foreach \i in {1,...,4}{ \node[vtx] at (g\i) {}; }
    \node[lbl] at (-6.82, 1.30) {$v_1$};
    \node[lbl] at (-4.18, 1.30) {$v_2$};
    \node[lbl] at (-4.18,-1.30) {$v_3$};
    \node[lbl] at (-6.82,-1.30) {$v_4$};
    \node[lbl] at (-5.5,-1.85) {$G$};

    \def\R{1.9}
    \draw[skel] (0,0) circle (\R);
    \coordinate (p1) at (90:\R);
    \coordinate (p2) at (45:\R);
    \coordinate (p3) at (0:\R);
    \coordinate (p4) at (-45:\R);
    \coordinate (p5) at (-90:\R);
    \coordinate (p6) at (-135:\R);
    \coordinate (p7) at (180:\R);
    \coordinate (p8) at (135:\R);

    \draw[circ=classicorange] (p5) to[bend right=42] (p6)
                              (p6) to[bend right=42] (p7)
                              (p7) to[bend right=80] (p5);
    \draw[circ=classicgreen] (p2) to[bend right=28] (p3);
    \draw[circ=classicgreen] (p2) to[bend right=52] (p3);
    \draw[circ=classicpurple] (p8) to[bend left=14]  (p4);
    \draw[circ=classicpurple] (p8) to[bend right=14] (p4);

    \foreach \i in {1,...,8}{ \node[vtx] at (p\i) {}; }

    \node[lbl] at (90:2.2)    {$v_1$};
    \node[lbl] at (45:2.2)    {$v_2$};
    \node[lbl] at (0:2.2)     {$v_3$};
    \node[lbl] at (-45:2.2)   {$v_4$};
    \node[lbl] at (-90:2.2)   {$v_2$};
    \node[lbl] at (-135:1.52) {$v_4$};
    \node[lbl] at (180:2.2)   {$v_1$};
    \node[lbl] at (135:2.2)   {$v_3$};
    \node[lbl, classicblue] at (112:2.4) {$S$};
    \node[lbl, classicorange] at (-135:2.85) {$C_1$};
    \node[lbl, classicgreen]  at (22:2.55) {$C_2$};
    \node[lbl, classicpurple] at (0.46,0.46) {$C_3$};
\end{tikzpicture}
\caption{%
    Left: Eulerian multigraph $G$ with four vertices. 
    Right: The Eulerian tour constructed by Hierholzer's algorithm. Here, $S$ is the ``skeleton'' of the tour, a spanning Eulerian tour, to which further ``detour circuits'' are attached. We process the graph in batches, such that these additional circuits are found with $O(n)$ working space. For clarity of the drawing, we show the tour $S$ ``expanded'', by drawing re-visited vertices several times; they are the same vertex in the output, though. Note that the original Hierholzer algorithm does not insist on $S$ being spanning.
}
\label{fig:skeleton-overview}
\end{figure}

%% file: bibliography.bib
@article{Hierholzer1873,
  title={{{\"U}ber die M{\"o}glichkeit, einen Linienzug ohne Wiederholung und ohne Unterbrechung zu umfahren}},
  author={Hierholzer, Carl and Wiener, Christian},
  journal={Mathematische Annalen},
  volume={6},
  number={1},
  pages={30--32},
  year={1873},
  publisher={Springer-Verlag Berlin/Heidelberg}
}

@article{PevznerTW01,
  title={{An Eulerian path approach to DNA fragment assembly}},
  author={Pevzner, Pavel A and Tang, Haixu and Waterman, Michael S},
  journal={Proceedings of the national academy of sciences},
  volume={98},
  number={17},
  pages={9748--9753},
  year={2001},
  publisher={National Academy of Sciences}
}

@article{Orloff74,
  title={{A Fundamental Problem in Vehicle Routing}},
  author={Orloff, Clifford S},
  journal={Networks},
  volume={4},
  number={1},
  pages={35--64},
  year={1974},
  publisher={Wiley Online Library}
}

@article{Yamamoto22,
  title={{A Novel Single-Stroke Path Planning Algorithm for 3D Printers using Continuous Carbon Fiber Reinforced Thermoplastics}},
  author={Yamamoto, Kohei and Luces, Jose Victorio Salazar and Shirasu, Keiichi and Hoshikawa, Yamato and Okabe, Tomonaga and Hirata, Yasuhisa},
  journal={Additive Manufacturing},
  volume={55},
  pages={102816},
  year={2022},
  publisher={Elsevier}
}

@article{GlazikSS23,
  title={{A one pass streaming algorithm for finding {Euler} tours}},
  author={Glazik, Christian and Schiemann, Jan and Srivastav, Anand},
  journal={Theory of Computing Systems},
  volume={67},
  number={4},
  pages={671--693},
  year={2023},
  publisher={Springer}
}

@article{HagerupKL19,
  title={Space-efficient {Euler} partition and bipartite edge coloring},
  author={Hagerup, Torben and Kammer, Frank and Laudahn, Moritz},
  journal={Theoretical Computer Science},
  volume={754},
  pages={16--34},
  year={2019},
  publisher={Elsevier}
}

@inproceedings{IsmailiPW26,
  title = {Space-Efficient {Hierholzer}: Eulerian Cycles in {$O(m)$} Time and {$O(n)$} Space},
  DOI = {10.1137/1.9781611978964.34},
  booktitle = {Symposium on Simplicity in Algorithms (SOSA)},
  publisher = {SIAM},
  author = {Ismaili Alaoui, Ziad  and Plump,  Detlef and Wild,  Sebastian},
  year = {2026},
  month = Jan,
  pages = {421–430}
}

@article{ColonUrschel2024,
  title = {Hamilton Powers of {Eulerian} Digraphs},
  volume = {31},
  ISSN = {1077-8926},
  url = {http://dx.doi.org/10.37236/11905},
  DOI = {10.37236/11905},
  number = {2},
  journal = {The Electronic Journal of Combinatorics},
  publisher = {The Electronic Journal of Combinatorics},
  author = {Colón,  Enrico and Urschel,  John},
  year = {2024},
  month = June 
}
